\documentclass[final,hidelinks,onefignum,onetabnum]{siamart220329}

\usepackage{lipsum}
\usepackage{amsfonts}
\usepackage{graphicx}
\usepackage{epstopdf}
\usepackage{algorithmic}

\usepackage{array}
\usepackage{float}
\usepackage{caption}
\usepackage{subcaption}

\usepackage{amssymb}

\ifpdf
  \DeclareGraphicsExtensions{.eps,.pdf,.png,.jpg}
\else
  \DeclareGraphicsExtensions{.eps}
\fi

\newsiamremark{remark}{Remark}
\newsiamremark{hypothesis}{Hypothesis}
\crefname{hypothesis}{Hypothesis}{Hypotheses}
\newsiamthm{claim}{Claim}

\newcommand{\halffunc}[1]{ 
	\left\{ \begin{array}{ll} #1 \end{array} \right. 
}

\newcommand{\Z}{{\mathbb{Z}}}
\newcommand{\V}{{\mathcal{V}}}
\newcommand{\E}{{\mathcal{E}}}
\newcommand{\sub}{\subseteq}		
\def\setsys{\mathfrak} 

\newcommand{\rem}[1]{}

\headers{Lattice Representations and Algorithms for Hypergraphs}
{Rawson, Myers, Green, Robinson, Joslyn}

\title{Formal Concept Lattice Representations and Algorithms for Hypergraphs}

\author{
Michael G. Rawson\footnote{Pacific Northwest National Laboratory; \{first\}.\{last\}@pnnl.gov} 
\and Audun Myers\footnotemark[1] 
\and Robert Green\footnote{University at Albany, SUNY; rgreen@albany.edu} \footnotemark[1] 
\and Michael Robinson\footnote{American University; michaelr@american.edu} 
\and Cliff Joslyn\footnotemark[1] \footnote{Binghamton University}
}

\usepackage{amsopn}

\ifpdf
\hypersetup{
  pdftitle={Formal Concept Lattice Representations and Algorithms for Hypergraphs},
  pdfauthor={Rawson et. al.}
}
\fi

\begin{document}

\maketitle

\begin{abstract}

There is increasing focus on  analyzing data represented as hypergraphs, which are better able to express complex relationships amongst entities than are graphs. 
Much of the critical information about hypergraph structure is available only in the intersection relationships of the hyperedges, and so forming the ``intersection complex'' of a hypergraph is quite valuable. This identifies a  valuable isomorphism between the intersection complex and the ``concept lattice'' formed from taking  the hypergraph's incidence matrix  as a  ``formal context'': hypergraphs also generalize graphs in that their incidence matrices are arbitrary Boolean matrices. 
This isomorphism allows connecting discrete algorithms for lattices and hypergraphs, in particular  $s$-walks or $s$-paths on hypergraphs can be mapped to order theoretical operations on the concept  lattice. We give new algorithms for formal concept lattices and hypergraph $s$-walks on concept lattices. We apply this to a large real-world dataset and find deep lattices implying high interconnectivity and complex geometry of hyperedges. 

\end{abstract}



\section{Introduction}
\label{sec:intro}

Binary relations are a  foundation of data science, and frequently come in the form of  large sparse Boolean matrices \cite{Ambrose_2020, Purvine2018, Robinson_langsec2022}. One  common interpretation of a Boolean matrix is as an incidence matrix of a hypergraph, wherein arbitrary sized groups of entities (rows) are joined together into various hyperedges (columns). The incident edges in a graph identify nodes they share in common. Moreover, intersecting hyperedges in a hypergraph identifies arbitrary numbers of nodes, and it is this intersection structure which carries all of the information coded in a hypergraph.
 
Section \ref{sec:back} gives mathematical preliminaries. In Section \ref{sec:theory}, we give correspondences between the concept lattice of a binary relation and the intersection complex of its hypergraph in Theorem \ref{thm:1}. In Section \ref{sec:algo}, we give new algorithms that use the correspondence to calculate the associated lattice (Algorithm \ref{alg:lattice_from_hypergraph}), the shortest $s$-path (Algorithm \ref{alg:shortest_s_path}), and the $s$-connected components (Algorithm \ref{alg:s_connected_components}). In Sections \ref{sec:example} and \ref{sec:case}, we demonstrate the theory and algorithms on example data and visualize the objects and their connections. Finally, in Section \ref{sec:conclusion} we discuss and conclude. 

\section{Background}
\label{sec:back}

A binary relation, $R$, over sets $X$ and $Y$ consists of ordered pairs $(x,y) \in R \sub X \times Y$ written $xRy$ for every $(x,y) \in R$  and $\neg (xRy)$ otherwise.
A relation can be described by a Boolean characteristic matrix $\chi \in \Z_2^{|X| \times |Y|}$ with 
\[ \chi[\kappa, \lambda] = \halffunc{
1,	& \text{if } x_\kappa R y_\lambda	\\
0,	& \text{if } \neg (x_\kappa R y_\lambda) \\
}
\]
for $\kappa$ and $\lambda$ in index sets $K$ and $\Lambda$ and for $x_\kappa \in X$ and $y_\lambda \in Y$.

A hypergraph is a pair $H=(V,E)$ where the set $V$ are vertices and $E$ is an indexed multiset of hyperedges $e \in E$ with $e \sub V$. Let set $\V$ index the vertices $V$ and $\E$ index the hyperedges $E$. 
We can then encode $H$ as a Boolean characteristic matrix $\chi \in \Z_2^{|\V| \times |\E|}$ with 
\[ \chi[\kappa, \lambda] = \halffunc{
1,	& \text{if } v_\kappa \in e_\lambda	\\
0,	& \text{if } v_\kappa \notin e_\lambda \\
}
\]
for $\kappa$ and $\lambda$ in index sets $\V$ and $\E$ and for $v_\kappa \in V$ and $e_\lambda \in E$. With this encoding, we express the hypergraph as the triple $H=(\V,\E,\chi)$. 

A formal context $C$ is defined by a triple $(G,M,\Gamma)$ where $G$ are  objects, $M$ are their attributes, and $\Gamma$ is a binary relation on $G \times M$ so that $g \Gamma m$ when object $g$ has attribute $m$ \cite{davey_priestley_2002}. 
We then call the pair $(A,B) \sub (G,M)$ a concept if 
\begin{enumerate}
    \item $A$ contains all of the objects that share the attributes in $B$ and 
    \item $B$ contains all of the attributes shared by the objects in $A$.
\end{enumerate}
The concept $(A,B)$ has extent $A$ and intent $B$. We define the tick operator that maps one to the other. For $A \sub G$ and $B \sub M$, define  
\begin{align*}
A' &= \{m \in M :  g \Gamma m \text{ for all } g \in A\},	\text{ and } \\
B' &= \{g \in G :  g \Gamma m \text{ for all } m \in B\}.
\end{align*}
A pair $(A,B)$ is a concept if and only if $A' = B$ and $B' = A$. Furthermore, $A=(A')'=A''$ and $B = (B')' = B''$. We call the set of all concepts $\setsys{B}(G,M,\Gamma)$.

\begin{definition}
  Let $P$ be a set. A \textbf{partial order} on $P$ is a binary relation $\leq$ which is reflexive, symmetric, and anti-transitive. 
  The pair $(P,\leq)$ is called a \textbf{partially ordered set} or a \textbf{poset}.

\end{definition}

\begin{definition}
  Given two partially ordered sets $(P, \leq_P)$ and $(Q,\le_Q)$, 
  a function $f: P \to Q$ is an \textbf{order embedding} if for all $x,y \in P$,
  it follows that $x \leq_P y$ if and only if $f(x) \le_Q f(y)$. 
    If $f$ is also surjective then it is an \textbf{order isomorphism}.
\end{definition}

\begin{definition}
A hypergraph $H = (\mathcal{V},\mathcal{E},\chi)=(V,E)$ is called topped if $V \in E$, and bottomed if $\emptyset \in E$. 
\end{definition}

\begin{definition}
  Given a hypergraph $H=(\mathcal{V},\mathcal{E},\chi)$, 
  its intersection complex $E^\cap \sub 2^V$ consists of all intersections of subsets of $E$:
  \[ E^\cap = \bigcup_{F \sub E, F \neq \emptyset} \left( \bigcap_{f\in F} f \right) .	\]
\end{definition}

The intersection complex $E^\cap$ is a partial order with $(E^\cap,\sub)$ ordering the intersections by subset, as is the 
edge set $(E,\sub)$ itself. Note that edge partial order $(E,\sub) \sub (E^\cap,\sub)$ is a sub-order of the intersection complex. An intersection complex is also called an intersection structure, intersection closure, or $\cap$-structure. 

Paths in graphs generalize to $s$-paths in hypergraphs. 

\begin{definition}  As in \cite{Aksoy2020, Hayashi}, for $s \in \Z_{> 0}$, an $s$-path in a hypergraph $H=(\mathcal{V},\mathcal{E},\chi)$ is a sequence of hyperedges $[e_0,e_1,...,e_k]$ where for all $0 \le i \le k-1, |e_i \cap e_{i+1}| \ge s$.
\end{definition}

\section{Theory} \label{sec:theory}

\begin{table}[h]
\begin{center}
\begin{tabular}{l | c | c | c | c | c | c | c }
 & Group 1 & Group 2 & Group 3 & Group 4 & Group 5 & Group 6 & Group 7 \\ 
\hline \hline
a & 0 & 1 & 1 & 1 & 0 & 0 & 0 \\ 
b & 1 & 1 & 0 & 1 & 0 & 0 & 0 \\
c & 1 & 1 & 0 & 0 & 0 & 0 & 0 \\
d & 0 & 1 & 1 & 0 & 0 & 0 & 0 \\
e & 1 & 0 & 0 & 0 & 1 & 0 & 0 \\
f & 0 & 0 & 0 & 0 & 1 & 1 & 0 \\
g & 0 & 0 & 0 & 0 & 1 & 1 & 1
\end{tabular}
\end{center}
\caption{Incidence matrix group memberships of $a,b,c,d,e,f,g$. } \label{tab:1}
\end{table}  

Consider a set of elements, $V=\{a,b,\ldots,g\}$, that are members of groups $E=\{1,2,\ldots,7\}$ as given by Table \ref{tab:1}. We take this first as an incidence matrix $\chi$ of a hypergraph $H =(\V,\E,\chi)$ with 
edge set 
    \[ E=\{\{b,c,e\},\{a,b,c,d\},\{a,d\},\{a,b\},\{e,f,g\},\{f,g\}, \{g\}\}  \]
as visualized in Figure \ref{fig:toy_example_HG}. The edge partial order is show in Figure \ref{fig:toy_example_HG_poset}.

\begin{figure}[htbp]
    \centering
    \begin{subfigure}{0.48\textwidth}
        \centering
        \includegraphics[width=.99\textwidth]{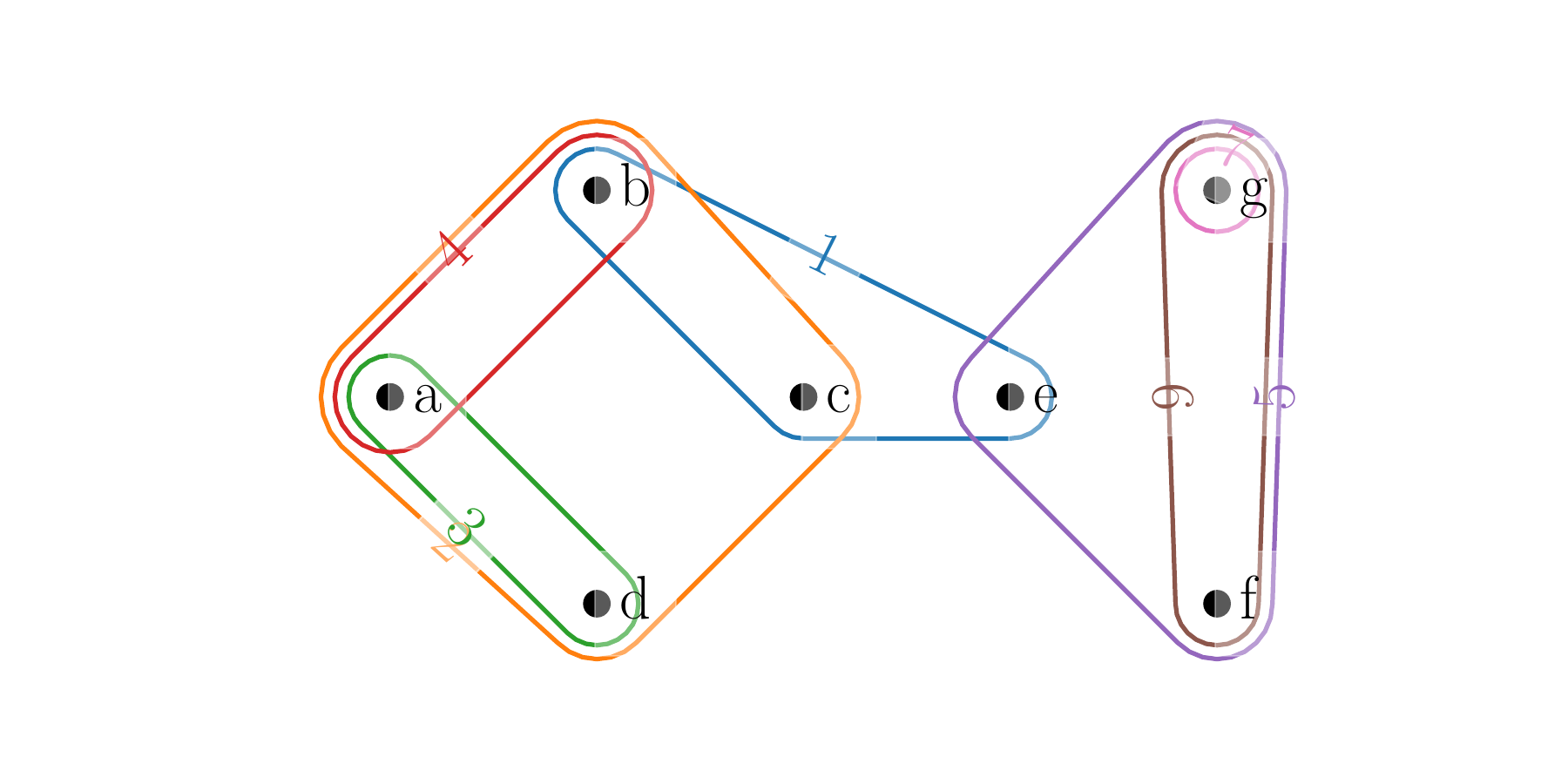}
        \caption{Hypergraph visualization.} 
        \label{fig:toy_example_HG}
    \end{subfigure}
    \hfill
    \begin{subfigure}{0.51\textwidth}
        \centering
        \includegraphics[width=0.99\textwidth]{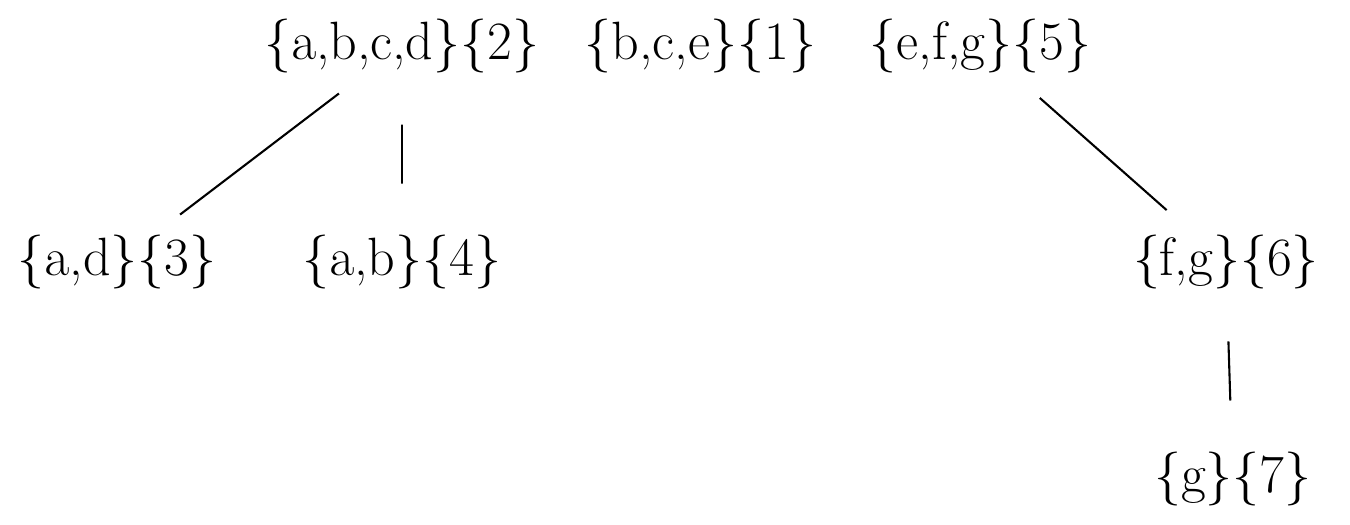}
        \caption{Edge partial order.}
        \label{fig:toy_example_HG_poset}
    \end{subfigure}
    \begin{subfigure}{0.7\textwidth}
        \centering
        \includegraphics[width=.7\textwidth]{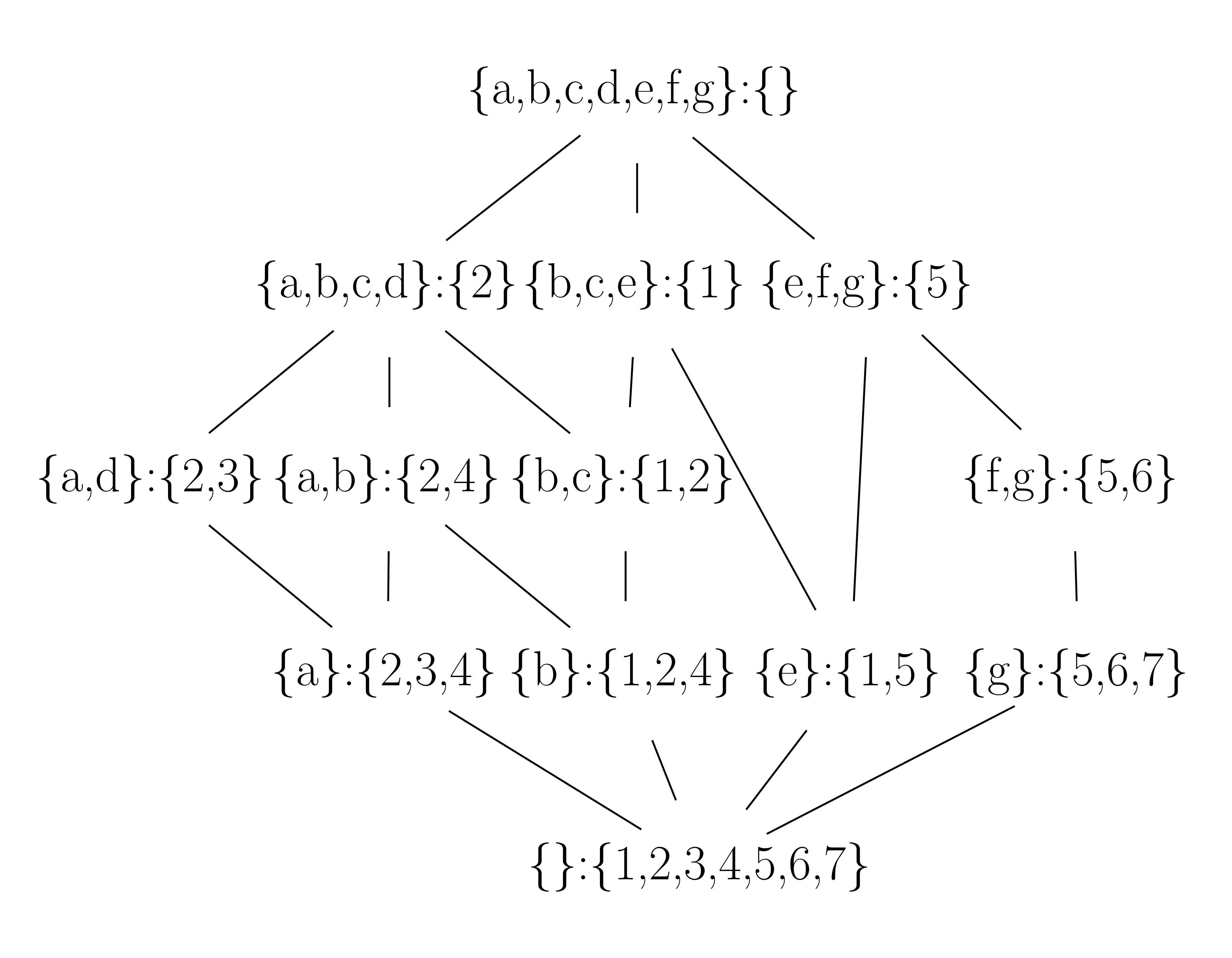}
        \caption{Concept lattice of incidence matrix.}
        \label{fig:toy_example_intersection_closure_lattice}
    \end{subfigure}
    \label{fig:toy_example}
    \caption{Hypergraph, edge partial order, and concept lattice of Table \ref{tab:1}. In Figure \ref{fig:toy_example_HG} vertices are dots and hyperedges are colored bands containing their vertices.}
\end{figure}

By taking $R$ as a formal context of $H$, for each $\kappa \in \V$ and $\lambda \in \E$, we set $\kappa R \lambda$ when $\chi[\kappa,\lambda]=1$. This yields formal concepts as shown in the concept lattice in Figure \ref{fig:toy_example_intersection_closure_lattice}, denoted as 

\begin{align*}
\mathfrak{B}(&X,Y,R) :=\{
(\{\},\{1,2,3,4,5,6,7\}),
(\{a\},\{2,3,4\}),
(\{b\},\{1,2,4\}),\\
&(\{g\},\{5,6,7\}),
(\{e\},\{1,5\}),
(\{a,b\},\{2,4\}),
(\{a,d\},\{2,3\}),
(\{b,c\},\{1,2\}),\\
&(\{f,g\},\{5,6\}), 
(\{b,c,e\},\{1\}),
(\{e,f,g\},\{5\}),
(\{a,b,c,d\},\{2\}),
(\{a,b,c,d,e,f,g\},\{\})
\}.
\end{align*}

The concept lattice gives a partial ordering to all of the elements of $\mathfrak{B}(X,Y,R)$. For example, $(\{a\},\{2,3,4\}) \le (\{a,d\},\{2,3\})$ since $\{a\} \subset \{a,d\}$ and $\{2,3,4\} \supset \{2,3\}$. 
We label the elements based on their inclusions using Galois notation \cite{davey_priestley_2002}. 
For example, hyperedge 7 is labeled as $\ell = \{g\}:\{5,6,7\}$ since the node $g$ is contained in the hyperedges 5, 6, and 7. 
Additionally, we can extract that this node is representative of hyperedge 7 by taking the immediate upstream nodes of $\{g\}:\{5,6,7\}$ which is $\{f,g\}:\{5,6\}$ and $\{e,f,g\}:\{5\}$ then taking the symmetric difference between the union hyperedge sets in the upstream labels and the hyperedge set of the node (i.e., $(\{5,6\} \cup \{5\}) \Delta \{5,6,7\} = \{7\}$). 
Surprisingly, the concept lattice maps to the intersection closure of the topped hypergraph obtained from $\chi$. 

\begin{theorem} \label{thm:1}
  For context $C=(X,Y,R)$, the intersection closure of the topped hypergraph obtained from $C$ is lattice isomorphic to $\mathfrak{B}(X,Y,R)$. 
\end{theorem}

\begin{proof}
    
Let hypergraph $H=(V,E)$ come from context $C$. So $V=X$ and $E \subset \mathcal P(V)$. Let $L$ be the intersection closure of the topped hypergraph. 

We will establish lattice isomorphism $\phi : \mathfrak{B}(X,Y,R) \rightarrow L$. Take a concept $c \in \mathfrak{B}(X,Y,R)$.
Let $\phi(c)=A$ where $c = (A,B)$ with $A \subset X$ and $B \subset Y$.
Write $\{v_1,...,v_k\}=A$.
and write $\{e_1,...,e_m\}=B$ as a subset of $E$.
Then $\{v_1,...,v_k\} \subset e_i$ for $1 \le i \le m$
and then $A \in L$. 
If $A=V$ then possibly $B=\emptyset$ however that case is covered since $L$ is topped and must contain $V$. Note $A$ may be $\emptyset$. 

So $\phi$ maps $\mathfrak{B}(X,Y,R)$ to $L$.
Clearly $\phi$ is injective.
Now, take an arbitrary $A \in L$ with $\{v_1,...,v_k\}=A$.
So $\{v_1,...,v_k\}=A = \cap_{e \in B} e $ for some $B \subset E$. Note $B$ may be $\emptyset$. 
Then we have $A = B'$
and $A' = B''$.
So $(B',B'') \in \mathfrak{B}(X,Y,R)$.
This implies $(B',B'') = (A,A')$.
Now let $c=(A,A')$.
Then $\phi^{-1}(A) = c \in \mathfrak{B}(X,Y,R)$
and $\phi$ is surjective.

We will show that for $c,d \in \mathfrak{B}(X,Y,R)$ and $c \le d$ implies $\phi(c) \le \phi(d)$.
Let $c = (A,B)$ and $d = (Q,R)$.
Then $A \subset Q$ and $B \supset R$.
So $\phi(c)=A$ and $\phi(d)=Q$
and $\phi(c) \subset \phi(d)$.
In the poset, $\phi(c) \le \phi(d)$.
A poset isomorphism between lattices is a lattice isomorphism \cite{davey_priestley_2002}.
$\blacksquare$

\end{proof}

The authors note similar theorems to the one above can be found in \cite{ayzenberg2019topology, Cattaneo2016}, but we find our formulation more concise.

\section{Algorithms} \label{sec:algo}
In this section we introduce three algorithms used for studying hypergraphs with lattice representations. The first is Algorithm~\ref{alg:lattice_from_hypergraph} for constructing the labeled intersection closure lattice from a hypergraph, the second is Algorithm~\ref{alg:shortest_s_path} to get the shortest $s$-path in the hypergraph using the lattice, and lastly is Algorithm~\ref{alg:s_connected_components} for calculating the $s$-connected components of a hypergraph using the lattice.
In Section~\ref{sec:example}, each of these algorithms is demonstrated on our example, which provides context for how these algorithms can be applied. 

\subsection{Computational Complexity}
Algorithm~\ref{alg:lattice_from_hypergraph} has been studied for a number of  years \cite{Butka, GaBWiR99, gonzalez2018modelling, KUZNETSOV,  lindig2000fast, wang2007algorithm}. The worst case runtime is $O(|V| \ 2^{2|V|})$ for the complete hypergraph $H=(V,E)$, see \cite{KUZNETSOV}. For a $k$-sparse hypergraph, $|E|=k$, the runtime can be reduced from $O(|V| \ 2^{2|V|})$ to $O(|V| \ 2^{2k})$. Where Algorithm~\ref{alg:lattice_from_hypergraph} is most useful is for $k$-lattice hypergraphs which is where the intersection lattice is limited to $k$ elements. In that case, the runtime is $O(k \ |V| |E|^2)$. 

\begin{algorithm}[h]
\caption{Calculate Intersection Lattice of Hypergraph}
\label{alg:lattice_from_hypergraph}
\begin{algorithmic}
\STATE{\textbf{Input:} Hypergraph $H=(V,E)$}
\STATE{\textbf{Output:} Lattice $L$}
\STATE{\textbf{Begin:}}
\STATE{$V_L = E$ \#Lattice elements/vertices}
\FOR{$k \in [1,min\{|V|,|E|\}]$}
    \FOR{$v_1,v_2 \in V_L$}
        \STATE{$V_L = V_L \cup (v_1 \cap v_2)$}
    \ENDFOR
\ENDFOR
\STATE{$V_L = V_L \cup V$ \#Add lattice top if missing}
\STATE{$E_L = \emptyset$ \#Lattice order relations}
\FOR{$v_1,v_2 \in V_L$}
    \IF{$v_1 \subset v_2$}
        \STATE{$E_L = E_L \cup (v_1,v_2)$}
    \ENDIF
\ENDFOR
\STATE \textbf{Return} $L=(V_L,E_L)$

\end{algorithmic}
\end{algorithm}

In practice, modern hardware has parallel abilities so we vectorize Algorithm \ref{alg:lattice_from_hypergraph} which we show in Algorithm \ref{alg:lattice_from_hypergraph_parallel}. This massively speeds up the runtime. The formal runtime for Algorithm \ref{alg:lattice_from_hypergraph_parallel} is not better but in practice the speed is effectively $O(k \ |E|)$ for small depth $k$ lattices. 

\begin{algorithm}[h]
\caption{Calculate Intersection Lattice of Hypergraph Vectorized}
\label{alg:lattice_from_hypergraph_parallel}
\begin{algorithmic}
\STATE{\textbf{Input:} Incidence Matrix $M : n\ \times\ p$}
\STATE{\textbf{Output:} Lattice $L$}
\STATE{\textbf{Begin:}}
\STATE{I = \{\} \#intersections set}
\STATE{N = M \#intersection incidence matrix}
\STATE{Remove duplicate columns in N}
\STATE{done = true}
\WHILE{not(done)}
    \STATE{dim\_N = dim(N)}
    \FOR{$i \in [1,dim(N,2)-1]$}

        \STATE{$S = M \wedge N[:,mod(i:p-1+i,dim(N,2))]$ \# intersection incidence}
        \STATE{intersections\_i = true\_indices(and(
                    $S$,\ dim=1))}
        \STATE{I = I $\cup$ 
                    \{(j,j-i) : j $\in$ intersections\_i \}}
        \STATE{S\_new\_indices = true\_indices(and(or(\\
            \quad broadcast(N, n $\times$ dim(N,2) $\times$ p) 
        $\veebar$ 
        broadcast(S,n $\times$ p $\times$ dim(N,2))$^T$, \\ 
            \quad dim=1),dim=2)) }
        \STATE{N = [N, S[:,intersections\_i $\cap$ S\_new\_indices]]}
    \ENDFOR
    \IF{dim\_N == dim(N)}
        \STATE{done = true}
    \ENDIF
\ENDWHILE
\STATE{C = \{\} \#containment set}
\FOR{$i \in [1,dim(N,2)-1]$}
    \STATE{containment\_i = true\_indices(not(and(
                $N \veebar (N \vee N[:,[i+1:dim(N,2), 1:i]])$,\ dim=1)))}
    \STATE{C = C $\cup$ 
                \{(j,j-i) : j $\in$ containment\_i \}}
\ENDFOR
\STATE{L\_element = N}
\STATE{L\_order = set(C)}
\STATE{L = (L\_element, L\_order)}
\end{algorithmic}
\end{algorithm}

\begin{algorithm}[h]
\caption{Shortest $s$-Path on Lattice}
\label{alg:shortest_s_path}
\begin{algorithmic}
\STATE{\textbf{Input:} \\
Integer $s$ \\
Labeled Lattice $L$ \\
source hyperedge $s_H$ \\ 
target hyperedge $t_H$}
\STATE{\textbf{Output:} \\
Lattice Path $P_L$ \\ 
Lattice Path Distance $d_L$ \\
Hypergraph Path $P_H$ \\ 
Hypergraph Path Distance $d_H$}
\STATE{\textbf{Begin:}}
\STATE{Define $L'$ as $L$ with top removed if top is not a hyperedge}
\STATE{Removes nodes in $L'$ with label $\ell_i = v_i : \varepsilon_i$ if $|v_i| < s$}
\STATE{Find $s_L$ as source node in lattice associated to $s_H$}
\STATE{Find $t_L$ as target node in lattice associated to $t_H$}
\STATE{Create $G_L$ as undirected graph version of $L'$}
\STATE{Apply depth-first search for shortest path in lattice between $s_L$ and $t_L$}
\STATE{Set $d_L = |P_L| - 1$}
\STATE{Initialize $d_H = 0$ and $P_H = []$}
\FOR{$n_i \in P_L$}
    \IF{$v_i$ of $n_i$ is a hyperedge}
        \STATE{Update $d_H++$ and $P_H$ by appending the hyperedge associated to $v_i$}
    \ENDIF
    \IF{$v_i^- \subset v_i \subset v_i^+$}
        \STATE{Update $d_H--$ and removing the hyperedge associated to $v_i$ from $P_H$}
    \ENDIF
\ENDFOR
\STATE \textbf{Return} $P_L$, $P_H$, $d_L$, $d_H$
\end{algorithmic}
\end{algorithm}

\begin{algorithm}[h]
\caption{$s$-Connected Components}
\label{alg:s_connected_components}
\begin{algorithmic}
\STATE{\textbf{Input:} $s$, Labeled Lattice $L$}
\STATE{\textbf{Output:} $s$-connected components in hypergraph}
\STATE{Define $L'$ as $L$ with top removed if top is not a hyperedge}
\STATE{Removes nodes in $L'$ with label $\ell_i = v_i : \varepsilon_i$ if $|v_i| < s$}
\STATE{Create $G_L$ as undirected graph version of $L'$}
\STATE{Get connected components $CC_L$ of $G_L$}
\STATE{Initialize $s$-connected components of hypergraph $CC_H$}
\FOR{$C_L$ in $CC_L$}
    \STATE{Initialize empty hypergraph component $C_H$}
    \FOR{$n_i$ in $C_L$}
        \IF{$v_i$ of $n_i$ is a hyperedge of $H$}
            \STATE{Update $C_H$ by including hyperedge $v_i$}
        \ENDIF
    \ENDFOR
    \STATE{Update $CC_H$ by including $C_H$}
\ENDFOR
\STATE \textbf{Return} {$CC_H$}
\end{algorithmic}
\end{algorithm}

\section{Example} \label{sec:example}

In this section, we use the example lattice in 
Figure~\ref{fig:toy_example_intersection_closure_lattice} to demonstrate the shortest $s$-path and the $s$-connected components in the lattice.

\begin{figure}[htbp]
    \centering
    \begin{subfigure}{0.55\textwidth}
        \centering
        \includegraphics[width=0.99\textwidth]{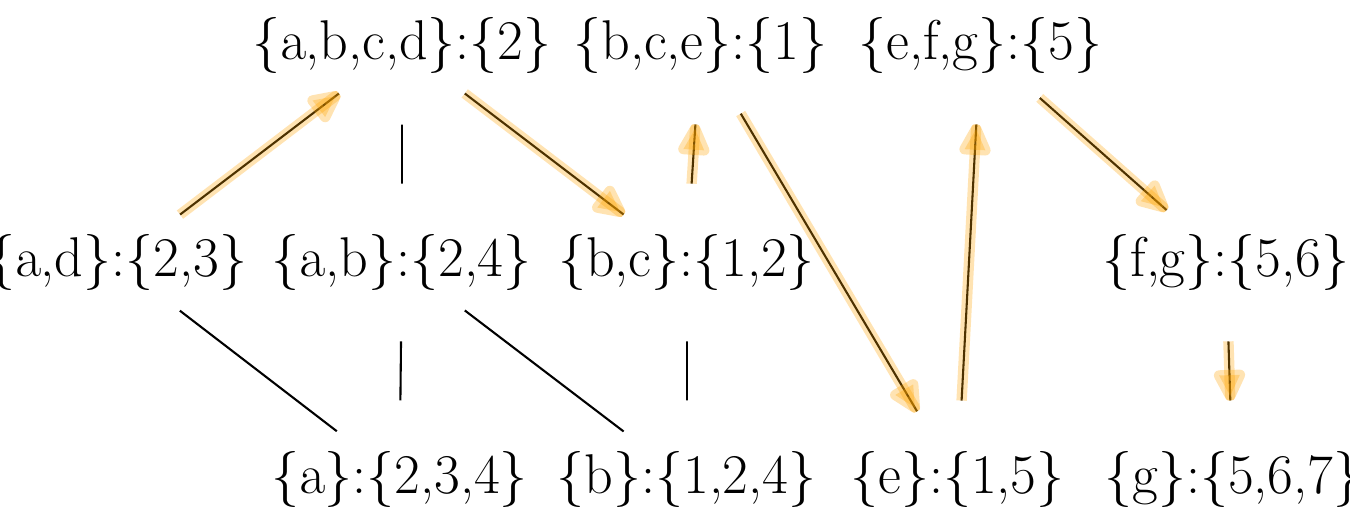}
        \caption{1-Path from hyperedge 3 to 7.}
        \label{fig:s1_walk_toy_example}
    \end{subfigure}
    \begin{subfigure}{0.55\textwidth}
        \centering
        \includegraphics[width=0.99\textwidth]{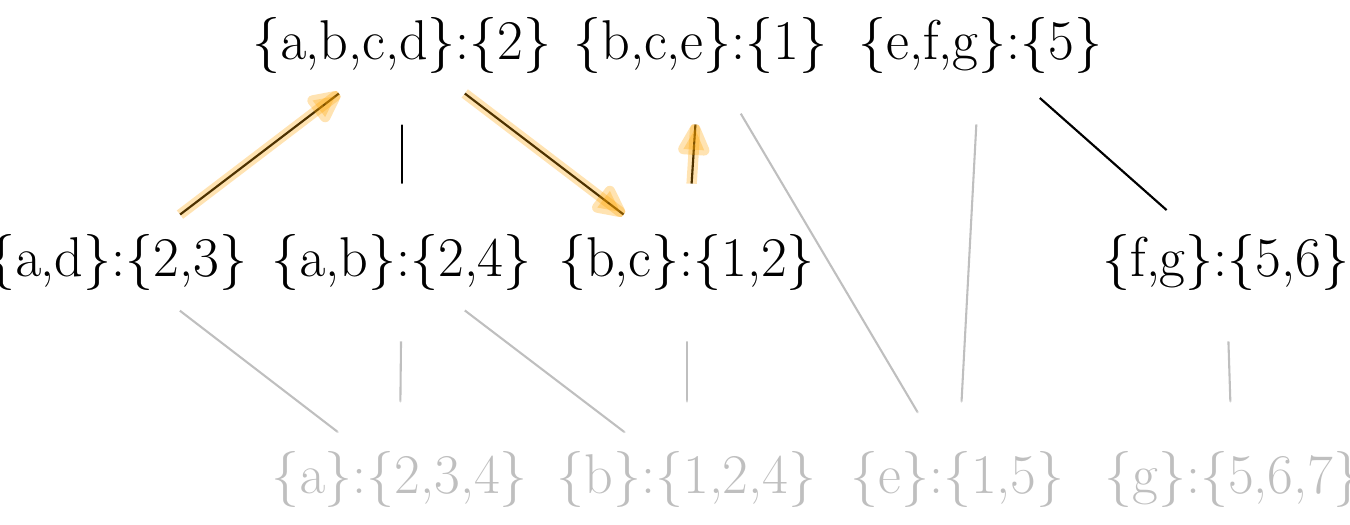}
        \caption{2-Path from hyperedge 3 to 1.}
        \label{fig:s2_walk_toy_example}
    \end{subfigure}
    \caption{Example shortest paths in the intersection closure lattice between hyperedges in the example in Figure~\ref{fig:toy_example_HG} for $s=1$ and $s=2$. For $s=2$, the unusable nodes and edges are translucent.}
    \label{fig:walks_toy_example}
\end{figure}

\subsection{Shortest $s$-Path}

We now show how our lattice representation of the hypergraph can be used to calculate the shortest $s$-path without the need for recalculating the lattice for each $s$.
A shortest $s$-path in an unweighted hypergraph between a source $e_s$ and target $e_t$ hyperedge is defined as a sequence of hyperedges $P_H = [e_s, e_1, e_2, \ldots, e_t]$ that minimizes the length of $|P_H|$, where each adjacent edge in $P_H$ must have an overlap of at least size $s \geq |e_i \cap e_{i+1}|$. 
This is typically calculated by first creating the $s$-line graph, which represents all of the hyperedges as as nodes and adds edges between these nodes if the hyperedges intersect with at least $s$ vertices in their hyperedges \cite{Aksoy2020}. 
It is typical to then apply a standard shortest path algorithm (e.g., depth-first search) on the $s$-line graph to get a sequence of nodes that are representative of hyperedges in the original hypergraph.
A drawback of using the $s$-line graph is that it requires that it is either recalculated for each $s$ or that the graph is weighted with the overlap information.
In the lattice we can also apply a shortest path algorithm; however we do not need to reconstruct the lattice for each desired $s$ as was done for the $s$-line graph. 

We now show two example shortest $s$-paths for $s=1$ and $s=2$ in Figure~\ref{fig:walks_toy_example} following Algorithm \ref{alg:shortest_s_path}.
In Figure~\ref{fig:s1_walk_toy_example}, we have highlighted the path in the intersection closure lattice for $s=1$ between hyperedges 3 and 7. We found the associated nodes in the lattice ($\{a,d\}:\{2,3\}$ for hyperedge 3 and $\{g\}:\{5,6,7\}$ for hyperedge 7) using the symmetric difference between hyperedges sets of the upstream node labels as previously described. 
The top is removed because that hyperedge did not exist in the hypergraph and the bottom since that node would require $s=0$. 
We found the shortest path in the lattice $P_L$ highlighted in orange in Figure~\ref{fig:s1_walk_toy_example} using a depth-first search. 
The path from 3 to 7 goes through both hyperedges and intersections with a length (and distance) in the lattice of $d_L = 7$. 
However, we are also interested in its equivalent interpretation in the hypergraph. 
To do this, we look at the sequence of nodes in $P_L$ and reduce the distance $d_L$ for each node that is associated to an intersection and for each sequence of three labels that are nested subsets (e.g., $\{e,f,g\}: \{5\}$ to $\{f,g\}: \{5,6\}$ to $\{g\}: \{5,6,7\}$) as outlined in Algorithm \ref{alg:shortest_s_path}. 
The resulting shortest $s$-path distance in the hypergraph is $d_H = 4$ with the sequence of hyperedges as $P_H = [3,2,1,5,7]$.

In a similar procedure as in Figure~\ref{fig:s1_walk_toy_example} we also found an $s=2$ shortest path between hyperedges $3$ and $1$. 
To demonstrate this, Figure~\ref{fig:s2_walk_toy_example} shows the nodes and connecting edges in the lattice that were removed based on their node sets being less than $s$ (e.g., intersection node $\{e\}:\{1,5\}$) as being semi-transparent. 
This shows, for example, that there is no 2-path from $3$ to $7$. However, there is still a path from $3$ to $1$ in the lattice with associated hypergraph 2-path $P_H = [3, 2, 1]$ and distance $d_H = 2$.

\subsection{$s$-Connected Components}
To find the $s$-connected components in a hypergraph using the concept lattice we build heavily from the procedure for finding the $s$-path. 
Specifically, we begin by removing nodes and connecting edges where the size of the node set for the lattice label is not at least of length $s$. 
Then, following Algorithm \ref{alg:s_connected_components}, we can take the components in the reduced poset $L'$ using a graph based method to get components in the lattice, called $C_L$. 
For each component in $C_L$ we simply get the hyperedges associated to each of the nodes in a component $C_L$ and reconstruct the hypergraph component $C_H$. 
For $s=1$ in our example, we would get one component that is our original hypergraph since the lattice is one component as shown in Figure~\ref{fig:s1_walk_toy_example}. 
For $s=2$ we get two components since their are two components in the lattice shown in Figure~\ref{fig:s2_walk_toy_example}.

\section{Case Study} \label{sec:case}
We next apply this new theory to a large real-world dataset. We experiment with a data set we call ``UKR14'' \cite{tracey-EtAl:2022:LREC1} which is a many gigabyte dataset of events extracted from news reports from the \textbf{2014} Russian invasion of Ukraine. These events can be viewed as a hypergraph where hyperedges are events (e.g.\ transport or attack events) with nodes in the hyperedges as the items involved in the event  (e.g., people and places). See Figure~\ref{fig:ukr14_HG} for a visualization of a subset of the dataset with edge and node labels replaced with numeric labels. We compute the formal concept lattice with Algorithm~\ref{alg:lattice_from_hypergraph} and plot it in Figure~\ref{fig:ukr14_poset}. In Figures 
\ref{fig:UKR14_dist_max_bottom},
\ref{fig:UKR14_dist_max_top},
\ref{fig:UKR14_dist_min_bottom}, and 
\ref{fig:UKR14_dist_min_top}, we plot the distributions of lattice nodes distances to root and leaf respectively. This gives us a distribution that indicates the shape of the lattice or how overlapping the events are. 
We see that there is significant depth up to 8 levels in the lattice, from Figure \ref{fig:UKR14_dist_min_top}. Note the logarithmic scale. 
This real world dataset follows a powerlaw distribution which can be modelled by the Chung Lu random graph model or the clustering variant \cite{Chung2000, Chung2023}. 
The maximal distances plotted in Figures \ref{fig:UKR14_dist_max_top} and \ref{fig:UKR14_dist_max_bottom} are much larger than the minimal ones. 
This tell us that the lattice and hence the concepts rarely form chains. 
As a hypergraph, chains are topologically contractible so nontrivial topology and geometry must be present which is an important invariant of the dataset. 

\begin{figure}[htb]
    \centerline{
    \begin{subfigure}[b]{0.5\textwidth}
        \centering
        \includegraphics[width=\textwidth]{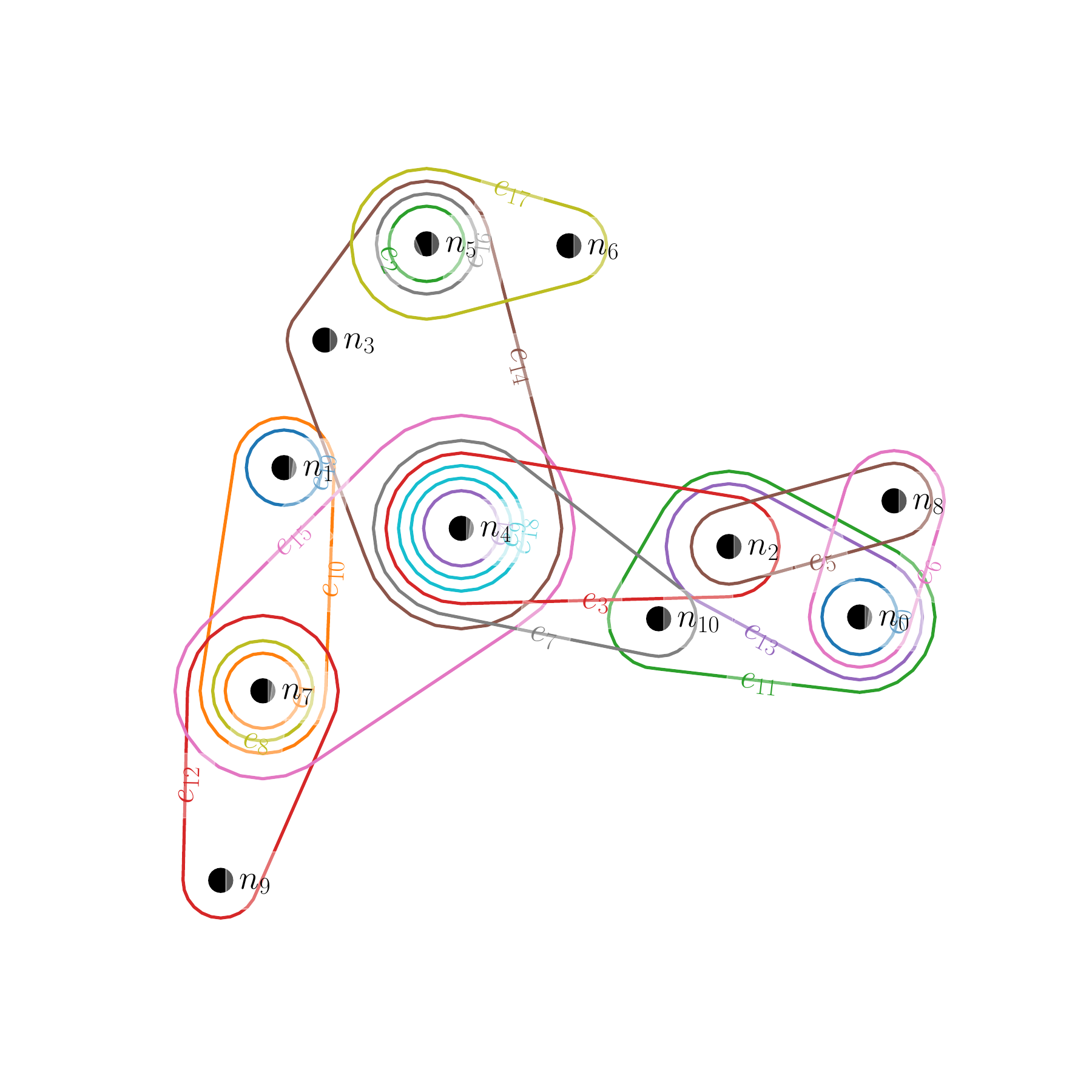}
        \caption{Hypergraph.}
        \label{fig:ukr14_HG}    
    \end{subfigure}
    \begin{subfigure}[b]{0.7\textwidth}
        \centering
        \includegraphics[width=\textwidth]{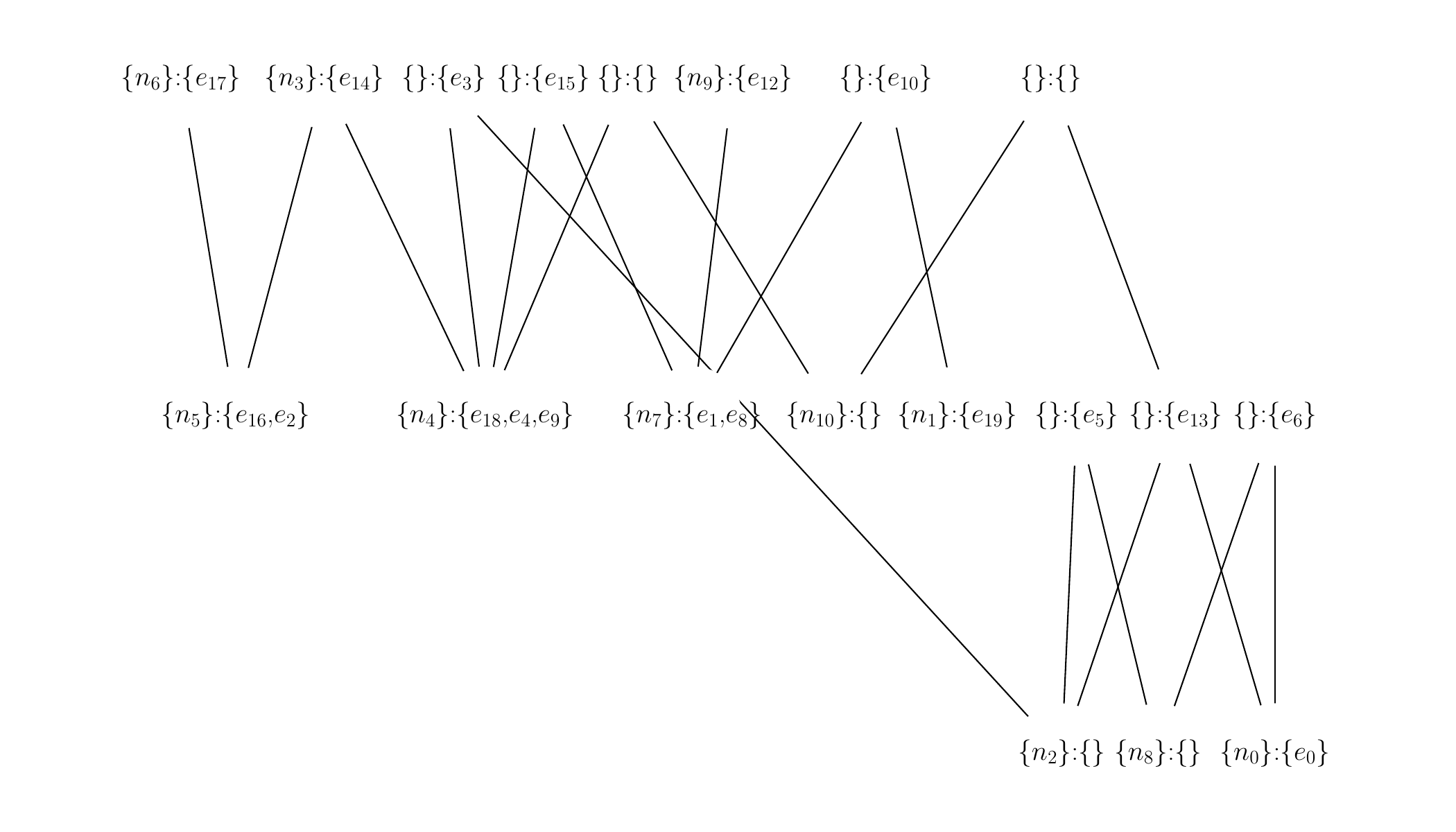}
        \caption{Lattice (lattice top and bottom omitted) of hypergraph.}
        \label{fig:ukr14_poset}
    \end{subfigure}
    }
    
    \caption{A small connected component of dataset UKR14.}
    \label{fig:ukr14}
\end{figure}

\begin{figure}[htb]
    \centerline{
    \begin{subfigure}[b]{0.5\textwidth}
        \centering
        \includegraphics[width=\textwidth]{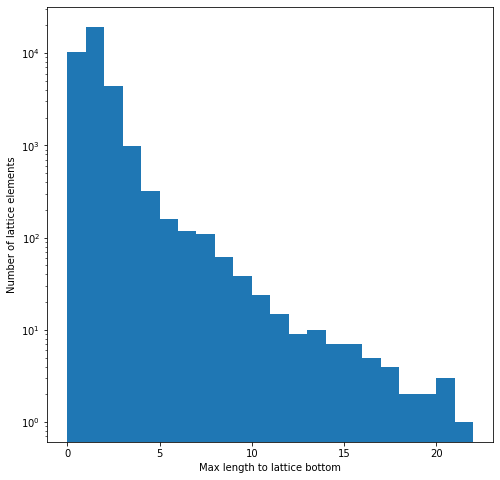}
        \caption{Maximal distances to bottom of lattice.}
        \label{fig:UKR14_dist_max_bottom}
    \end{subfigure}
    \begin{subfigure}[b]{0.5\textwidth}
        \centering
        \includegraphics[width=\textwidth]{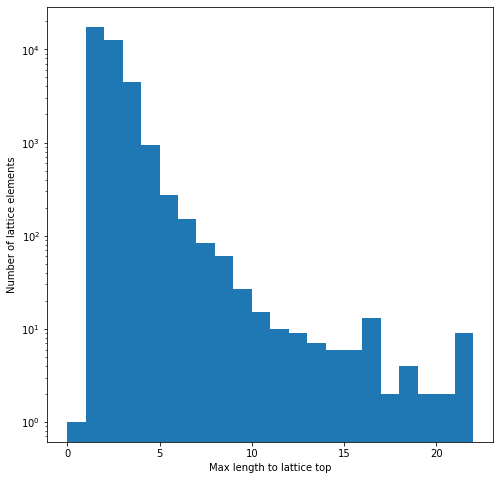}
        \caption{Maximal distances to top of lattice.}
        \label{fig:UKR14_dist_max_top}
    \end{subfigure}}
    
    \caption{Histogram of maximal distances to top/bottom of lattice of the largest connected component of dataset UKR14.}
    \label{fig:UKR14_dist1}

\end{figure}
    
\begin{figure}[htb]
    \centerline{
    \begin{subfigure}[b]{0.5\textwidth}
        \centering
        \includegraphics[width=\textwidth]{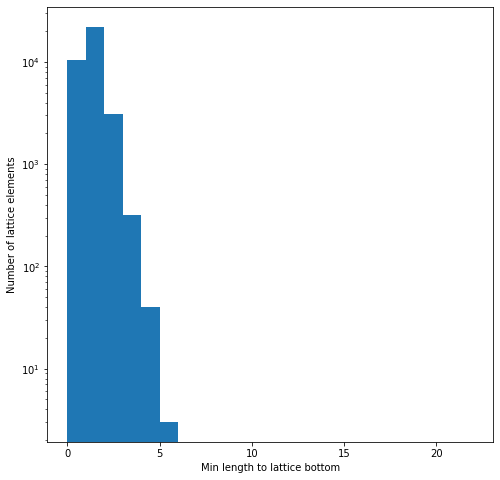}
        \caption{Minimal distances to bottom of lattice.}
        \label{fig:UKR14_dist_min_bottom}
    \end{subfigure}
    \begin{subfigure}[b]{0.5\textwidth}
        \centering
        \includegraphics[width=\textwidth]{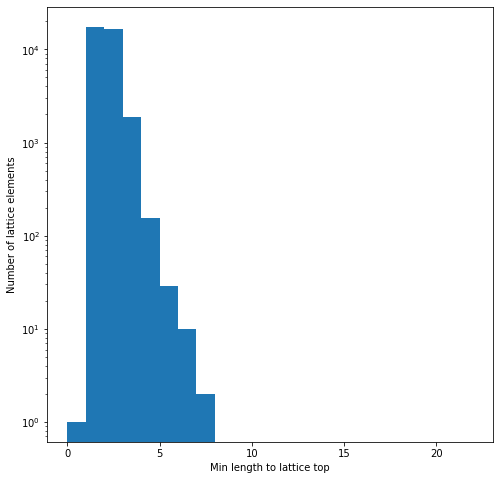}
        \caption{Minimal distances to top of lattice.}
        \label{fig:UKR14_dist_min_top}
    \end{subfigure}}
    
    \caption{Histogram of minimal distances to top/bottom of lattice of a connected component of dataset UKR14.}
    \label{fig:UKR14_dist2}
\end{figure}

\section{Conclusion} \label{sec:conclusion}

We gave a direct proof that the concept lattice and topped intersection closure hypergraphs are lattice isomorphic.
Hypergraph $s$-paths and lattice $s$-paths are mutually determined.
We provided a new vectorized algorithm to calculate the concept lattice. We gave an algorithm for hypergraph $s$-paths via lattice and applied it to get $s$-connected components from the lattice.
Our algorithms can also be extended to $s$-clustering by pruning hyperedges in the hypergraph.
Our future work is to compute lattices, paths, and clusters for larger hypergraphs and further analyze computational complexity.

\section*{Acknowledgements}

Pacific Northwest National Laboratory\footnote{Released under PNNL-SA-187287.} is a multiprogram
national laboratory operated for the US Department
of Energy (DOE) by Battelle Memorial Institute
under Contract No. DE-AC05-76RL01830.
Robinson was partially supported by the Defense Advanced Research Projects Agency (DARPA) SafeDocs program under contract HR001119C0072.  Any opinions, findings and conclusions or recommendations expressed in this material are those of the authors and do not necessarily reflect the views of DARPA.

\bibliographystyle{siamplain}
\bibliography{ref}
\end{document}